\documentclass[11pt]{amsart}
\usepackage{amsmath,amssymb,amscd}
\usepackage{amsthm}
\usepackage{xspace} 
\usepackage[dvips]{epsfig}
\usepackage{graphics}
\usepackage{vaucanson-g}
\usepackage{tikz}
\usepackage{mathtools}

\usepackage{url}

\newtheorem{proposition}{Proposition}
\newtheorem{corollary}{Corollary}
\newtheorem{theorem}{Theorem}
\newtheorem{lemma}{Lemma}

\theoremstyle{remark}

\newtheorem{example}{Example}

\def\ldots{\mathinner{\ldotp\ldotp}}       

\newcommand{\In}{\operatorname{In}}
\newcommand{\Out}{\operatorname{Out}}

\def\Z{\mathbb{Z}}
\def\A{\mathcal{A}}
\def\B{\mathcal{B}}

\def\G{\mathcal{G}}
\def\H{\mathcal{H}}

\def\P{\mathcal{P}}

\def\T{\mathcal{T}}
\def\X{\mathsf{X}}

\newcommand{\resp}{{resp.}\xspace }
\ChgStateLineWidth{0.5}
\ChgEdgeLineWidth{0.5}
\FixVCScale{0.4}

\title{Finite-type-Dyck shift spaces}
\author{Marie-Pierre B\'eal}
\address{Universit\'e Paris-Est, Laboratoire d'informatique Gaspard-Monge, UMR 8049 CNRS}
\email{beal@univ-mlv.fr}

\author{Michel Blockelet}
\address{Universit\'e Paris-Est, Laboratoire d'Algorithmique,
  Complexit\'e et Logique}
\email{michel.blockelet@univ-mlv.fr}

\author{C\v{a}t\v{a}lin Dima}
\address{Universit\'e Paris-Est, Laboratoire d'Algorithmique, Complexit\'e et Logique}
\email{catalin.dima@u-pec.fr}
\thanks{This work is supported by the French National Agency (ANR) through "Programme d'Investissements d'Avenir" (Project ACRONYME $\text{n}^\circ$ANR-10-LABX-58) and through the ANR EQINOCS}


\date{\today}
\begin{document}

\begin{abstract}
We study some basic properties of sofic-Dyck shifts and
finite-type-Dyck shifts. We prove that the class of sofic-Dyck shifts
is stable under proper conjugacies. We prove a Decomposition Theorem of a proper conjugacy between edge-Dyck shifts
into a sequence of Dyck splittings and amalgamations.
\end{abstract}

\maketitle

\tableofcontents

\section{Introduction}

Shifts (or subshifts) of sequences are defined as sets of bi-infinite sequences
of symbols over a finite alphabet avoiding a given set of finite factors (or blocks) called forbidden
factors \cite{LindMarcus1995}.
In \cite{BealBlockeletDima2013}, we defined the notions of
sofic-Dyck shifts and finite-type-Dyck shifts which extend the notion of
Markov-Dyck shifts introduced by Krieger and Matsumoto (see \cite{Krieger2006}, \cite{Inoue2006}, \cite{InoueKrieger2010},\cite{Matsumoto2011b}, \cite{Matsumoto2011c}, 
\cite{KriegerMatsumoto2011b},\cite{Krieger2012}).
Sofic-Dyck shifts are shifts of sequences whose set of forbidden (or
allowed) blocks is a visibly pushdown language
\cite{BealBlockeletDima2013}. Visibly pushdown languages
\cite{AlurMadhusudan2004} is a strict subclass of context-free languages 
which is closed by intersection and complementation.

A sofic-Dyck shift is accepted 
by a finite-state automaton (or a labelled graph) equipped with a graph
semigroup, over an alphabet which is partitioned into three disjoint
sets of symbols, the call symbols, the return symbols, and internal symbols
(for which no matching constraints are required). Such automata are
called Dyck automata.
Finite-type Dyck shifts are accepted by Dyck automata which are local (or definite).

In this paper, we study some basic properties of these classes of
context-free shifts. We introduce the notion of \emph{proper
 block map}. A proper block map $\Phi$ is a block map
between two shifts over two three-type alphabets such that $\Phi(x)=y$
implies $y_i$ and $x_i$ have the same type for any integer $i$ Roughy
speaking, a call
(\resp return, internal) 
symbol is mapped to a call symbol (\resp return, internal).
We show that a subshift is a sofic-Dyck shift if and only if it is the
proper factor of a finite-type-Dyck shift and that the class of sofic-Dyck shifts
is stable under proper conjugacies.

We define two notions of in and out state-splitting map, in and out
state-amalgamation map, together with the notions
of in-split, out-split, in-amalgamation, and out-amalgamation
Dyck automaton. Amalgamation and splitting maps are proper
conjugacies. The first notion is the classical notion of state splittings. The second one, called trim splitting,
allows one to remove some edges or matchings which are not essential.

We define the notion of edge-Dyck shifts. They play the same role as edge shifts for the sofic-Dyck class. An edge-Dyck shift is
defined by a Dyck graph whose edges are partitioned into three types
of edges: call edges,
return edges and internal edges. 
We prove a Decomposition Theorem for edge-Dyck shifts. We show that
two edge-Dyck shifts are conjugate through a proper conjugacy if and only
there is a sequence of splittings and amalgamations allowing to
transform one Dyck graph into the other one.
We use trim splittings for the final step.
Since
trim in-splittings do not commute (as classical in-splittings), the
result does not allow us to derive a decision process for the proper
conjugacy of one-sided edge-Dyck shifts (see \cite{Kitchens1998} or
\cite{LindMarcus1995}) for
the notion of one-sides shifts of sequences.

\section{Sofic-Dyck shifts and finite-type-Dyck shifts} \label{section.dyckshift}

\subsection{Shifts}
We briefly introduce below some basic notions of symbolic dynamics. We refer
to \cite{LindMarcus1995, Kitchens1998} for an introduction to this theory.
Let $A$ be a finite alphabet. 
The \emph{shift transformation} $\sigma$ on $A^\Z$ is defined by 
\begin{equation*}
\sigma((x_i)_{i \in \Z} = (x_{i+1})_{i \in \Z}, 
\end{equation*}
for $(x_i)_{i \in \Z} \in A^\Z$.

A \emph{subshift}  (or \emph{shift}) of $A^\Z$ is a closed
shift-invariant subset of $A^\Z$ equipped with the product of the
discrete topology. If $X$ is a shift, a finite word
is a \emph{block of} $X$ if it appears as a factor of some bi-infinite
sequence of $X$.  We denote by $\B(X)$ the set of factors (or blocks) of $X$ and by
$\B_n(X)$ the set of blocks of length $n$ of $X$. 
If $x= (x_i)_{i \in I} $ is a word and $i,j \in I$ with $i \leq
j$, we denote by $x[i,j]$ the factor $x_i x_{i+1} \cdots x_j$ of $x$.
Let
$F$ be a set of finite words over the alphabet $A$. We denote by
$\X_F$ the set of bi-infinite  sequences of $A^\Z$ avoiding each word
of $F$. The set $\X_F$ is a shift and any shift is the set
of bi-infinite  sequences avoiding each word of some set of finite
words. When $F$ can be chosen finite (\resp regular), the shift  $\X_F$ is called a \emph{shift of
  finite type} (\resp \emph{sofic}).
When $F$ can be chosen context-free, the shift  $\X_F$ is called a
\emph{context-free shift}.


   Let $X \subset A^\Z  ,Y \subset B^\Z$ be subshifts and $m,a$
   be nonnegative integers. A map $\Phi: X \rightarrow Y$ is called
   an $(m,a)$-\emph{local map} (or an $(m,a)$-\emph{block map}) if there
   exists a function $\phi : \mathcal{B}_{m+a+1}(X) \rightarrow B$ such
   that, for all $x \in X$ and any $i \in \Z$, $\Phi(x)_i= \phi(x_{i-m} \dotsm x_{i-1}x_ix_{i+1}
   \dotsm x_{i+a})$. A \emph{block map} is a map which
   is an $(m,a)$-block map for some nonnegative integers $(m,a)$. 
   The function $\phi$ is called a \emph{local function} associated
   to $\Phi$ and $(m,a)$ and $\Phi(X)$ is said to be a \emph{factor}
   of $X$.

Let $X$ be a shift over the alphabet $A$, the \emph{higher-block
  shift} of order $n$ is the shift denoted $X^{[n]}$ over $B=A^n$ 
defined as the image of $X$ by an $(m,a)$-block map such that $m+a+1=n$
and whose local function is
the map $\phi: A^n \rightarrow B$ which is the identity map
over $A^n$. This map is a conjugacy. The higher-block shift of order $n$ is also called the shift of sequences of overlapping blocks
of lengths $n$ of $X$.

 
\subsection{Dyck automata} \label{section.definitions}

We define the notion of Dyck automata and sofic-Dyck shift introduced in
\cite{BealBlockeletDima2013}. 

We consider an alphabet $A$ which is a disjoint union of three finite
sets of letters $(A_c,A_r,A_i)$. The sets $A_c$, $A_r$ and $A_i$ are called the
\emph{call alphabet}, the
\emph{return alphabet}, and the \emph{internal (or local) alphabet}
respectively.

Let $\A=(Q,E,A)$ be a directed
labelled graph (or automaton) on a finite alphabet $A$ with a finite set of vertices
$Q$ and a (finite) set of edges
$E \subset Q \times B \times Q$.  
We say that $\A$ is \emph{deterministic} if there is at most one edge
with a given label and a given starting state. 

Let $M$ be a set of pairs $(p,a,q),(r,b,s)$ of edges of $\A$
with $a \in A_c$ and $b \in A_r$. It is called the set of
\emph{matched edges}.
We define the \emph{graph semigroup} $S$ associated to $(\A,M)$ as the semigroup
generated by the set $E \cup \{x_{pq}
  \mid p,q \in Q\} \cup \{0\}$ equipped with the following relations.
\begin{align*}
0s=s0 &= 0  &\text { for } &s \in S,\\
x_{pq}x_{qr} &= x_{pr}   &\text { for } &p,q,r \in Q,\\
x_{pq}x_{rs} &= 0  &\text { for } &p,q,r,s \in Q, q\neq r,\\
(p,\ell,q) & = x_{pq} &\text { for } &p,q, \in Q, \ell \in A_i,\\
(p,a,q)x_{qr}(r,b,s) &= x_{ps}  &\text { for }
&((p,a,q),(r,b,s)) \in M,
\end{align*}
\begin{align*}
(p,a,q)x_{qr}(r,b,s) &= 0 &\text { for } &a \in A_c, b \in A_r, ((p,a,q),(r,b,s)) \notin M,\\
(p,a,q)(r,b,s) &= 0, &\text { for } &p,q,r,s \in Q, q \neq r, a,b \in A,\\
x_{pp} (p,a,q) = (p,a,q) &= (p,a,q) x_{qq} &\text { for } &p,q \in Q, a
\in A,\\
x_{pq} (r,a,s) = 0 &= (r,a,s) x_{tu} &\text { for } &p,q \in Q, a
\in A, \: q \neq r, \: s \neq t.\\
\end{align*}
Note that $X=(x_{pq})_{p,q \in Q}$ satisfies $X^2=X$ assuming that
$x_{pq} + x_{pq} = x_{pq}$. 

If $\pi$ is a finite path of $\A$, we denote by 
$f(\pi)$ its image in the graph semigroup $S$.
A finite path $\pi$
of $\A$ such that $f(\pi) \neq 0$ is said to be an \emph{admissible
  path} of $(\A,M)$ (or of $\A$ when $M$ is understood).
A finite word is \emph{admissible} for $(\A,M)$ if it is the label of some
admissible path of $(\A,M)$.
A bi-infinite path is \emph{admissible} if all its finite factors are admissible.
A path $\pi$ such that $f(\pi)=x_{pq}$ is called a \emph{Dyck path}
going from $p$ to $q$. It is called a \emph{prime Dyck path} if it has
no strict prefix which is also a Dyck path.

The \emph{sofic-Dyck shift} $X \subset A^\Z$ is defined as the set of
labels of admissible
bi-infinite paths of $(\A,M)$. 
The pair $(\A,M)$ is called a \emph{Dyck automaton} over $A$
and the shift is denoted by $\X_{(\A,M)}$.

Note that there may exists in $\A$ several bi-infinite paths having the
same bi-infinite label.
We say that the sofic-Dyck shift is \emph{accepted} (or \emph{presented})
by the Dyck automaton $(\A,M)$.

The \emph{full-Dyck shift} over the alphabet $A=(A_c,A_r,A_i)$,
denoted $\X_A$,
is the shift of all sequences accepted by the one-state Dyck automaton
$(\A=(Q=\{p\},E,A),M)$ containing each loop $(p,a,p)$ for $a \in A$, and where each
edge $(p,a,p)$ is matched with each edge $(p,b,p)$ when $a \in A_c, b \in
A_r$. Hence $\X_A$ is the set of all sequences over $A_r\cup A_c \cup
A_i$.

For $m,a$ nonnegative integers, we define the $(m,a)$-Dyck-De-Bruijn automaton
as the Dyck automaton $(\A=(Q=A^m \times A^a,E,A),M)$ containing the edges
$((au,bv),b,(ub,vc))$ and where each edge $((au,bv),b,(ub,vc))$ is
matched with $((a'u',b'v'),b',(u'b',v'c'))$ when $b \in A_c$ and $b' \in
A_r$. Any $(m,a)$-Dyck-De-Bruijn automaton over the alphabet
$A=(A_c,A_r,A_i)$ accepts the full-Dyck shift over $A$. We denote by $\In(p)$ (\resp
$\Out(p)$ the set of edges coming in $p$ (\resp going out of $p$).

A Dyck automaton is \emph{normalized} if any label of a finite
admissible path is a block of the Dyck shift accepted the automaton.
This means that for any finite word $u$ labeling a finite admissible path $\pi$ of the
automaton there is a bi-infinite path labelled by a word $x$ containing $u$ as
factor. Note that this does not mean that $\pi$ itself is extensible to a
bi-infinite admissible path. It is shown in \cite{BealBlockeletDima2013} that
any sofic-Dyck shift is accepted by a normalized Dyck automaton.

Let $(\G=(Q,E),M)$ be a Dyck automaton.
An edge which belongs to a bi-infinite
admissible path is called \emph{essential}.
A pair of edges $(e,f)$, where $e$ is a call edge and $f$ a return edge,
for which there is a
bi-infinite admissible path $(p_i \rightarrow p_{i+1})_{i \in \Z}$
such that $e=(p_0 \rightarrow p_1)$, $f= (p_j \rightarrow p_{j+1})$
for some $j >0$ and $(p_i \rightarrow p_{i+1})_{1 \leq i \leq j-1}$
is a Dyck path or is empty, is called an \emph{essential matched pair}.
A Dyck automaton is \emph{essential} if each edge is essential and each matched
pair of edges is essential. 

The following proposition was proved in \cite{BealBlockeletDima2013}.
\begin{proposition} 
The set of blocks of a sofic-Dyck shift is a visibly-pushdown
language. Conversely, any factorial extensible visibly-pushdown
language is the set of blocks of a sofic-Dyck shift.
\end{proposition}

Let $(\A=(Q,E,A),M)$ and  be $(\A'=(Q',E',A'),M')$ be two
Dyck automata. We define their product as the
Dyck automaton $(\A,M) \times (\A',M')= ((Q \times Q', F,B),N)$ over
$B$
where
\begin{itemize}
\item $B=(A_c \times A'_c, A_r \times A'_r, A_i \times A'_i)$,
\item $((p,p'),(a,a'),(q,q'))$ belongs to $F$ if and only if
  $(p,a,q)$ belongs to $E$ and $(p',a',q')$ belongs to $E'$.
\item $(((p_1,p'_1),(a_1,a'_1),(q_1,q'_1)),
  ((p_2,p'_2),(a_2,a'_2),(q_2,q'_2)))$ 
belongs to $N$ if and only if 
$((p_1,a_1,q_1),(p_2,a_2,q_2))$ belongs to $M$ 
and $((p'_1,a'_1,q'_1),$ $(p'_2,a'_2,q'_2))$ belongs to $M'$. 
\end{itemize}

\begin{proposition} \label{proposition.localproduct}
The intersection of two sofic-Dyck shifts accepted by $(\A,M)$ and
$(\B,N)$ is a sofic-Dyck shift accepted by the product of $(\A,M)$ and $(\B,N)$.
The product of two local Dyck automata is a local Dyck automaton. 
\end{proposition}
\begin{proof}
The proof is straightforward.
\end{proof}

\begin{proposition}
Let $X,Y$ be two sofic-Dyck (\resp finite-type-Dyck) shifts accepted by $(\A,M)$ and $(\B,N)$
respectively, $X \cap Y$ is a sofic-Dyck (\resp finite-type-Dyck)
shift accepted by the product $(\A,M) \times (\B,N)$.
\end{proposition}

\begin{proof}
The proof follows from Proposition~\ref{proposition.localproduct}. The
first part 
is also a consequence of the known fact that the intersection of two
visibly pushdown languages is a visibly pushdown language (see for
instance \cite{AlurMadhusudan2004}).
\end{proof}

Let $(\A,M)$ be a Dyck automaton over
$A=(A_c,A_r,A_i)$. Let $m,a$
   be nonnegative integers.
We says that $\A$ is $(m,a)$-\emph{local} (\resp $(m,a)$-\emph{weak-local}  if whenever
two paths (\resp two admissible paths)  of length $m+a$
$(p_i,a_i,p_{i+1})_{-m \leq i \leq a-1}$, $(q_i,a_i,q_{i+1}) _{-m \leq i \leq a-1}$,
of $\A$ have the same label, then $p_0=q_0$.
We says that $\A$ (or $(\A,M)$) is \emph{local} if it is $(m,a)$-local
for some nonnegative integers $m$ and $a$.

Note that if $\A$ is deterministic local if there is a
nonnegative integer $m$ such that for each word $x$ of length $m$ of $\A$,
all paths of $\A$ labelled by $x$ end in a same state
(depending on $x$).

A \emph{finite-type-Dyck shift} is a sofic-Dyck shift presented by a
weak local Dyck automaton. As shown in Proposition
\ref{proposition.presentation}, it is also presented by a local Dyck automaton.

\begin{proposition} \label{proposition.presentation}
If $X$ is a finite-type-Dyck shift, then there are nonnegative
integers $m,a$ such that $X$ is accepted by an $(m,a)$-local
Dyck automaton.
\end{proposition}
\begin{proof}
Let $X$ accepted by an $(m,a)$-weak-local Dyck automaton
$(\A,M)$ over $A=(A_c,A_r,A_i)$. Let $(\B=(R,F,A),N)$ be the Dyck automaton over $A$
defined by 
\begin{itemize}
\item $R$ is the set of admissible paths of length $m+a$ of $(\A,M)$,
\item $((p_i,a_i,p_{i+1})_{-m \leq i\leq a-1},a_0,
  (p_i,a_i,p_{i+1})_{-m+1 \leq i\leq a}) \in F$ if $(p_i,a_i,p_{i+1})$ 
belongs to $E$ for $-m \leq i\leq a$,
\item The two edges $((p_i,a_i,p_{i+1})_{-m \leq i\leq a-1},a_0,
  (p_i,a_i,p_{i+1})_{-m+1 \leq i\leq a})$ and 
$((p'_i,a'_i,p'_{i+1})_{-m \leq i\leq a-1},a'_0
  (p'_i,a_i,p'_{i+1})_{-m+1 \leq i\leq a})$ are matched in $N$
 if $(p_0,a_0,p_1)$ and $(p'_0,a'_0,p'_1)$ are matched in $N$.
\end{itemize}
The Dyck automaton $(\B,N)$ accepts $X$. Let us show that $(\B,N)$
is $(2m,2a)$-local. Let $(\pi_i)_{-2m \leq i \leq 2a-1},(\pi'_i)_{-2m
  \leq i \leq 2a-1}$ be two (possibly non-admissible) paths of
length $2m+2a$ in $(\B,N)$ which have the same label.
Let $\pi_i=(p_{i+j},a_{i+j},p_{i+j+1}) _{-m \leq j
    \leq a-1}$ and $\pi'_i=(p_{i+j},a_{i+j},p_{i+j+1})_{-m \leq j
    \leq a-1}$. Since each path $\pi_i$ or $\pi'_i$ is an admissible
  path of $(\A,N)$, we have $p_i=p'_i$ for $-m \leq i \leq a-1$ and
  thus $\pi_0 = \pi'_0$.
\end{proof}

\begin{figure}[htbp]
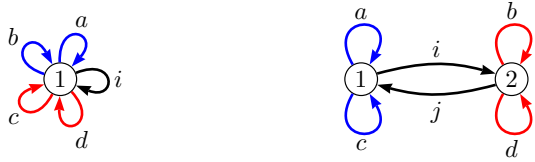

    \centering
\FixVCScale{0.5}
\VCDraw{%
\begin{VCPicture}{(0,0)(16,4)}
\MediumState
\State[1]{(2,2)}{1}
\SetEdgeArrowWidth{8pt}
\SetEdgeLineWidth{2pt}
\SetEdgeLineColor{blue}
\CLoopL[0.5]{70}{1}{a}
\CLoopL[0.5]{140}{1}{b}
\SetEdgeLineColor{red}
\CLoopL[0.5]{-140}{1}{c}
\CLoopL[0.5]{-70}{1}{d}
\SetEdgeLineColor{black}
\CLoopL[0.5]{0}{1}{i}
\VCPut{(8,-1)}{
\MediumState
\State[1]{(2,3)}{11}
\State[2]{(6,3)}{22}
\SetEdgeLineColor{blue}
\LoopN[0.5]{11}{a}
\LoopS[0.5]{11}{c}
\SetEdgeLineColor{red}
\LoopN[0.5]{22}{b}
\LoopS[0.5]{22}{d}
\SetEdgeLineColor{black}
\ArcL[0.5]{11}{22}{i}
\ArcL[0.5]{22}{11}{j}
}
\end{VCPicture}%
        }
        \caption{The full-Dyck shift $\X_A$ over $A=(A_c,A_r,A_i)$ with $A_c=\{a,b\}$,
          $A_r=\{c,d\}$ and $A_i=\{i\}$ (on the left). Each edge
          labelled by some letter in $A_c$ is matched with each edge
          labelled in $A_r$. On the right is pictured a
          Dyck automaton $(\A',M)$ over $A'=(A_c,A_r,A'_i=\{i,j\})$
          where the edge labelled by $a$ is matched with the 
      edge labelled by $b$, and the edge labelled by $c$ is matched with the 
     edge labelled by $d$. It
          accepts a finite-type-Dyck shift which is
          not a full-Dyck shift.}\label{figure.FTDyck}
\end{figure}

\begin{example}
The full-Dyck shift $\X_A$ over $A=(A_c,A_r,A_i)$ is a finite-type-Dyck shift.
It is the set of bi-infinite sequences presented
by the automaton $(\A,M)$ of Figure~\ref{figure.FTDyck}.
On the right part of this figure is pictured an $(1,0)$-local Dyck automaton
accepting a finite-type-Dyck shift.
\end{example}

The class of finite-type-Dyck shifts (FT-Dyck) is a strict subclass of sofic-Dyck
shifts and it strictly contains the class of finite-type shifts (SFT)
as is shown in Figure~\ref{figure.patates}.

\newcommand{\RectA}{(-4,-2) rectangle (7,2)}
\newcommand{\RectB}{(-3,-1.8) rectangle (3,1.8)}
\newcommand{\RectC}{(-2.8,0) rectangle (2.8,1.7)}
\newcommand{\RectD}{(-2.6,-1.6) rectangle (-0.5,1.6)}

\begin{figure}[htbp]
    \centering
\begin{tikzpicture} [rounded corners]
\fill[color=blue!5] \RectA;
\fill[color=blue!10] \RectB;
\fill[color=blue!20,opacity=0.5]  \RectC;
\fill[color=blue!30,opacity=0.5]\RectD;
\draw \RectA \RectB \RectC \RectD ;
\draw (1,0.7) node{sofic shifts};
\draw (1.2,-0.8) node{sofic-Dyck shifts};
\draw (5,-1) node{context-free shifts};
\draw (-1.6,0.8) node{SFT};
\draw (-1.6,-0.8) node{FT-Dyck};
\end{tikzpicture}
\caption{Subclasses of the context-free shift classes.} \label{figure.patates}
\end{figure}
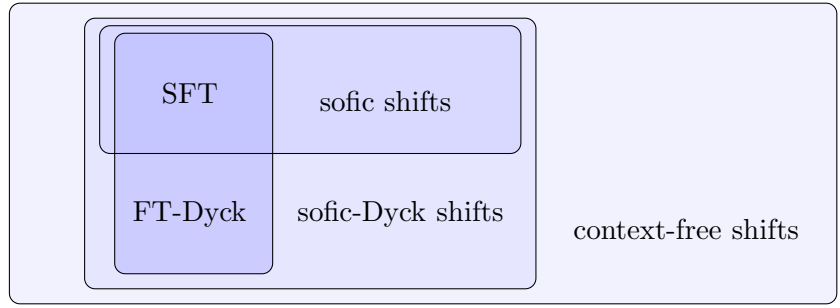

\subsection{Proper block-maps}

A block map $\Phi: \X_A \rightarrow \X_{A'}$, where $A=(A_c,A_r,$ $A_i)$
and $A'=(A'_c,A'_r,A'_i)$,
is called \emph{proper} if 
$\Phi(x)_j \in A'_c$ (\resp $A'_r,$ $A'_i$) whenever $x_j \in A_c$
(\resp $A_r$, $A_i$) for any $j \in \Z$.

\begin{proposition} \label{proposition.image}
A subshift is a sofic-Dyck shift if and only it is the proper factor
of a finite-type-Dyck shift.
\end{proposition}

\begin{proof}
Let $(\A=(Q,E,A),M)$ be a Dyck automaton over $A=(A_c,A_r,A_i)$
accepting a sofic-Dyck shift $X$. 

Let $(\B=(Q,F,E),N)$ be the
$(1,0)$-local Dyck automaton over $E=(E_c,E_r,E_i)$ where $E_c$
(\resp $E_r$, $E_i$) is the
set of edges labelled by $A_c$ (\resp $A_r$, $A_i$). The edges of $\B$
are $(p,e,q)$ for $e=(p,a,q)\in E$. Two edges $(p,e,q)$ and
$(p',e',q')$ are matched in $N$  if and only if $e$ and $e'$ are
matched in $M$.
Let $\Phi$ be the
$(0,0)$-block map with local function $\phi(p,a,q)=a$. We have
$\Phi(\X_{(\B,N)})=X$. Hence $X$ is a factor of a finite-type-Dyck shift.

Conversely, suppose that $X$ is a shift space for which there is a 
finite-type-Dyck shift $Y$ and a proper block map $\Phi$ from $Y$ onto
$X$. Suppose that $\Phi$ has memory $m$ and anticipation $a$ with a
local function $\phi$. By increasing $m$ if necessary, we can assume
that $Y$ is accepted by an $(m,a)$-local Dyck automaton 
$(\A=(Q,E,A),M)$ over $A=(A_c,A_r,A_i)$. 

We define the Dyck automaton $(\T=(R,F,B),N)$ (called a transducer)
with $R = Q \times A^{m}\times A^{a}$, $B=(A_c \times A'_c, A_r \times
A'_r, A_i \times A'_i)$. 

An edge $((p,(au,bv),$ $(b,\phi(aubvc)),(q,(ub,vc)))$,
with $a,b,c \in A$, belongs to $F$ if and only if $(p,b,q)$ belongs to
$E$ and  $p$ (\resp $q$) is the state $p_0$ (\resp $q_0$) of any path labelled by
$aubv$ (\resp $ubvc$) going from $p_{-m}$ to $p_a$ 
(\resp going from $q_{-m}$ to $q_a$). 

A pair 
$((p,(au,bv),(b,\phi(aubvc)),(q, (ub,vc)))$ 
$((r,(dw,ez)),(e,\phi(dwezf)),$ $(s, (we,zf)))$ of $F$
with $a,b,c,d,e,f \in A$, belongs to $N$ if and only if 
$((p,b,q),$ $(r,e,s))$ belongs to $M$.
The Dyck automaton obtained from $(\T,N)$ by discarding the second
component of all labels of edges of $\T$ is an $(m,a)$-local
Dyck automaton accepting $Y$. The image $X$ of $Y$ by $\Phi$ is 
accepted by Dyck automaton obtained from $(\T,N)$ by discarding the
first component of all labels of edges of $\T$. Hence $X$ is sofic-Dyck.
\end{proof}

\begin{corollary} \label{proposition.imageSofic}
A proper factor of a sofic-Dyck shift is a sofic-Dyck shift.
\end{corollary}
\begin{proof}
Suppose that $\Phi$ is a proper block map from $Y$ onto $X$
and that $Y$ is a sofic-Dyck shift. By Proposition
\ref{proposition.image}, there is a proper block map $\Psi$ from $Z$ onto $Y$
with $Z$ a finite-type-Dyck shift. Since $\Phi \circ \Psi$ is a proper
block map, we conclude again by Proposition
\ref{proposition.image} that $X$ is sofic-Dyck.
\end{proof}

\begin{proposition} \label{proposition.higherblock}
If $X$ is a sofic-Dyck (\resp finite-type-Dyck) shift, then the higher-block shifts of
$X$ are sofic-Dyck (\resp finite-type-Dyck) shifts.
\end{proposition}
\begin{proof}
The proposition is a consequence of Proposition~\ref{proposition.image}.
\end{proof}

\subsection{Dyck state-splitting} \label{section.splitting}

In this section, we define two notions of state splitting for Dyck automata and sofic-Dyck shifts:
state-splitting and trim state-splitting.

The notion of state-splitting is an extension of the notion of state splitting
for sofic shifts.
Let $(\A=(Q,E,A),M)$ be a Dyck automaton over $A=(A_c,A_r,A_i)$.
Let $p \in Q$ and $\P$ a partition $(\P_1,\ldots,\P_k)$ of size $k$ of
the edges coming in $p$. We define
a Dyck automaton $(\A'=(Q',E',A),M')$ by 
\begin{itemize}
\item $Q'= Q \setminus \{p\} \cup \{p_1,\ldots,p_k\}$,
\item $(q,a,r) \in E'$ if $q,r \neq p$ and $(p,a,r) \in E$,
\item $(q,a,p_i) \in E'$ for each $1 \leq i \leq k$ such that
 $(q,a,p) \in \P_i$,
\item $(p_i,a,r) \in E'$ 
for each $1 \leq i \leq
  k$ such that  $(p,a,r) \in E$.
\item $M'$ is the set of pairs of edges $(q,a,r),(s,b,t)$ where $a \in
  A_r, b \in A_c$ such that
  $(\pi(q),a,\pi(r)),(\pi(s),b,\pi(t))$ $\in M$
where $\pi(q) = q$ for $q \neq p$ and $\pi(p_i)=p$ for $1 \leq i \leq k$.
\end{itemize}
This automaton $(\A',M')$ is called a \emph{Dyck in-split} of
$(\A,M)$ and $(\A,M)$ is called a \emph{Dyck in-amalgamation} of
$(\A',M')$. If the classes of the partition $\P$ are singletons, the
Dyck in-split graph (\resp map) is called a \emph{complete Dyck
  in-split} graph (\resp map).
The notions of \emph{Dyck out-split} 
and \emph{Dyck out-amalgamation}
are defined similarly. 

Note that, if $(\A,M)$ is a Dyck amalgamation of $(\A',M')$, then
the map $\pi: E' \rightarrow E$ where $\pi(p,a,q)=
(\pi(p),a,\pi(q))$ defines a $(0,0)$-block conjugacy $(\A',M')$ onto $(\A,M)$.

\begin{figure}[htbp]
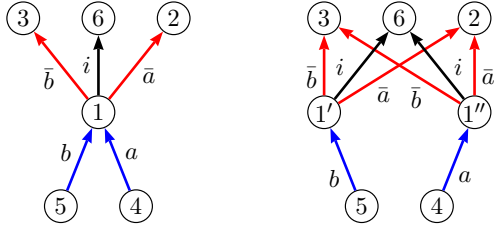

    \centering
\FixVCScale{0.5}
\VCDraw{%
\begin{VCPicture}{(0,-1)(12,5)}
\MediumState
\State[1]{(2,2)}{1}
\State[2]{(4,4.5)}{2}
\State[6]{(2,4.5)}{6}
\State[3]{(0,4.5)}{3}
\State[4]{(3,-0.5)}{4}
\State[5]{(1,-0.5)}{5}
\SetEdgeArrowWidth{8pt}
\SetEdgeLineWidth{2pt}
\SetEdgeLineColor{red}
\EdgeR[0.5]{1}{2}{\bar{a}}
\EdgeL[0.5]{1}{3}{\bar{b}}
\SetEdgeLineColor{blue}
\EdgeR[0.5]{4}{1}{a}
\EdgeL[0.5]{5}{1}{b}
\SetEdgeLineColor{black}
\EdgeL[0.5]{1}{6}{i}
\VCPut{(8,0)}{
\MediumState
\State[1']{(0,2)}{11}
\State[1'']{(4,2)}{00}
\State[2]{(4,4.5)}{22}
\State[6]{(2,4.5)}{66}
\State[3]{(0,4.5)}{33}
\State[4]{(3,-0.5)}{44}
\State[5]{(1,-0.5)}{55}
\SetEdgeLineColor{red}
\EdgeR[0.3]{11}{22}{\bar{a}}
\EdgeL[0.3]{11}{33}{\bar{b}}
\EdgeR[0.3]{00}{22}{\bar{a}}
\EdgeL[0.3]{00}{33}{\bar{b}}
\SetEdgeLineColor{blue}
\EdgeR[0.3]{44}{00}{a}
\EdgeL[0.3]{55}{11}{b}
\SetEdgeLineColor{black}
\EdgeL[0.3]{11}{66}{i}
\EdgeR[0.3]{00}{66}{i}
}
\end{VCPicture}%
        }
        \caption{A Dyck state-splitting of the state $1$ into $1'$ and
          $1"$. The alphabet is
          $A=(\{a,b\},\{\bar{a},\bar{b}\},\{i\})$. The edges labelled by $a$ (\resp $b$) are matched with
          the edges labelled by $\bar{a}$ (\resp $\bar{b})$.
}\label{figure.properStateSplitting}
\end{figure}


The notion of trim state splitting is defined on Dyck automata is the following.
Let $(\A=(Q,E,A),M)$ be a Dyck automaton over $A=(A_c,A_r,A_i)$.
Let $p \in Q$ and $\P$ a partition $(\P_1,\ldots,\P_k)$ of size $k$ of
the edges coming in $p$. We define
a Dyck automaton $(\A'=(Q',E',A),M')$ by 
\begin{itemize}
\item $Q'= Q \setminus \{p\} \cup \{p_1,\ldots,p_k\}$,
\item $(q,a,r) \in E'$ if $q,r \neq p$ and $(p,a,r) \in E$,
\item $(q,a,p_i) \in E'$ for each $1 \leq i \leq k$ such that
 $(q,a,p) \in \P_i$, 
\item $(p_i,a,r) \in E'$ 
for each $1 \leq i \leq
  k$ such that  $(p,a,r) \in E$.
\item $M'$ is the set of pairs of edges $(q,a,r),(s,b,t)$ where $a \in
  A_r, b \in A_c$ such that
  $(\pi(q),a,\pi(r)),(\pi(s),b,\pi(t))$ $\in M$
where $\pi(q) = q$ for $q \neq p$ and $\pi(p_i)=p$ for $1 \leq i \leq
k$.
\item Edges $(p_i,a,r)$ which are not essential in $(\A',M')$ are
  removed from $E'$. Matched pairs $(q,a,r),(p_i,b,t)$ or $(p_i,b,t),(q,a,r)$
  which are not essential are removed from $M'$.
\end{itemize}
This automaton $(\A',M')$ is called a \emph{trim Dyck in-split} of
$(\A,M)$ and $(\A,M)$ is called a \emph{trim Dyck in-amalgamation} of
$(\A',M')$. Note that a trim Dyck in-split automaton is an essential Dyck automaton.
Since removing unessential edges or matched pairs does not affect the
bi-infinite admissible paths of a Dyck automaton, the
map $\pi: E' \rightarrow E$ defined by $\pi(p,a,q)=
(\pi(p),a,\pi(q))$ defines a $(0,0)$-block conjugacy from 
$(\A',M')$ onto $(\A,M)$.

\begin{figure}[htbp]
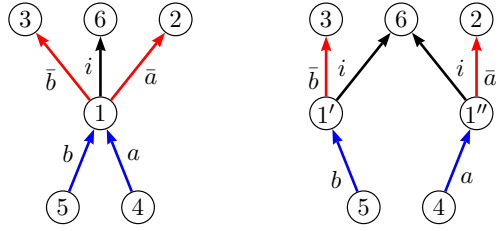

    \centering
\FixVCScale{0.5}
\VCDraw{%
\begin{VCPicture}{(0,-1)(12,5)}
\MediumState
\State[1]{(2,2)}{1}
\State[2]{(4,4.5)}{2}
\State[6]{(2,4.5)}{6}
\State[3]{(0,4.5)}{3}
\State[4]{(3,-0.5)}{4}
\State[5]{(1,-0.5)}{5}
\SetEdgeArrowWidth{8pt}
\SetEdgeLineWidth{2pt}
\SetEdgeLineColor{red}
\EdgeR[0.5]{1}{2}{\bar{a}}
\EdgeL[0.5]{1}{3}{\bar{b}}
\SetEdgeLineColor{blue}
\EdgeR[0.5]{4}{1}{a}
\EdgeL[0.5]{5}{1}{b}
\SetEdgeLineColor{black}
\EdgeL[0.5]{1}{6}{i}
\VCPut{(8,0)}{
\MediumState
\State[1']{(0,2)}{11}
\State[1'']{(4,2)}{00}
\State[2]{(4,4.5)}{22}
\State[6]{(2,4.5)}{66}
\State[3]{(0,4.5)}{33}
\State[4]{(3,-0.5)}{44}
\State[5]{(1,-0.5)}{55}
\SetEdgeLineColor{red}
\EdgeL[0.3]{11}{33}{\bar{b}}
\EdgeR[0.3]{00}{22}{\bar{a}}
\SetEdgeLineColor{blue}
\EdgeR[0.3]{44}{00}{a}
\EdgeL[0.3]{55}{11}{b}
\SetEdgeLineColor{black}
\EdgeL[0.3]{11}{66}{i}
\EdgeR[0.3]{00}{66}{i}
}
\end{VCPicture}%
        }
        \caption{A trim Dyck state-splitting of the state $1$ into $1'$ and
          $1"$. The alphabet is
          $A=(\{a,b\},\{\bar{a},\bar{b}\},\{i\})$. The edges labelled by $a$ (\resp $b$) are matched with
          the edges labelled by $\bar{a}$ (\resp $\bar{b})$. }\label{figure.properStateSplitting}
\end{figure}

\begin{proposition}
The in-split (\resp out-split) of a Dyck automaton accepts the same
sofic-Dyck shift. The in-split (\resp out-split) of an $(m,a)$-local Dyck
automaton
is an $(m+1,a)$-local (\resp $(m,a+1)$-local) Dyck automaton.
\end{proposition}


\subsection{Characterization of finite-type-Dyck shifts}

In this section we give an intrinsic characterization of finite-type-Dyck shifts.

Let $A=(A_c,A_r,A_i)$ be an alphabet. Let $m,a$ be nonnegative
integers, $F$ be a (finite) set of words of
length $m+a+1$ over $A$, and $G$ be a (finite) set of pairs 
$(u_{-m}\cdots u_{a},v_{-m}\cdots v_{a})$ of words
of length $m+a+1$ such that $u_0 \in A_c$, $v_0 \in A_r$.
We define a semigroup $S(F,G)$ generated by the set $\{(au,ub) \mid u \in A^{m+a}, a,b \in A \}  \cup \{x_{u',v'}
  \mid u',v'  \in A^{m+a}\} \cup \{0\}$ with the following relations.
\begin{align*}
0s=s0 &= 0  &\text { for } & s \in S,\\
x_{u'v'}x_{v'w'} &= x_{u'w'}   &\text { for } &u',v',w' \in A^{m+a},\\
x_{u'v'}x_{r's'} &= 0  &\text { for } &u',v',r',s' \in A^{m+a}, v\neq r,\\
(au,ub)
& = 0 &\text { if\phantom{x} }   &aub \in F,\\
(au,ub)
& = x_{au,ub} &\text { if\phantom{x} }   & aub \notin F, b \in A_i,\\
(au,ub)x_{ub,cv}(cv,vd) &= x_{au,vd} &\text { for }
&b \in A_r, d \in A_c,\\
&   &
& ((au,ub),(cv,vd)) \notin G,\\ 
(au,ub)x_{bd,cv}(cv,vd) &= 0 &\text { for }
&((au,ub),(cv,vd)) \in G,\\ 
(au,ub)(cv,vd) &= 0 &\text { for } &ub \neq vc,\\
x_{au,au}(au,ub) = (au,ub) &=  (au,ub) x_{ub,ub}, &&\\
x_{u',u'}(au,ub) = 0 &= (au,ub) x_{v',v'} &\text { for } &
u' \neq au, v' \neq ub.
\end{align*}
The image $f(w)$ of a finite word $w$ of length greater than or equal to
$m+a$ in $S(F,G)$ is defined as the product of its overlapping
blocks of length $m+a+1$ in $S(F,G)$.
A word $w \in A^{\geq m+a+1}$ is said to be \emph{admissible} for $S(F,G)$ if and only if 
$f(w) \neq 0$. 

A language $L$ of finite words
over $A$ is said be \emph{strictly-locally-Dyck testable} if and only
there are nonnegative integers $m,a$, 
finite sets $F,G$ as above such that, if $w$ has a length at least
$m+a+1$, then $w \in L$ if and only if $w$ is admissible for $S(F,G)$.

\begin{proposition}
A sofic-Dyck shift is a finite-type-Dyck shift if and only if its
set of blocks is strictly-locally-Dyck testable.
\end{proposition}

\begin{proof}
If $X$ is a finite-type-Dyck shift over $A$, it is accepted by an 
$(m,a)$-local Dyck automaton $(\A,M)$ for some nonnegative
integers $m,a$. Let $F$ be the set of word of $A^{m+a+1}$ which are
not blocks of $X$ and $G$ be the set of pairs of words $(u_{-m}\cdots u_a,v_{-m}\cdots v_a)$
 of length $m+a+1$ such that $u_0 \in A_c$, $v_0 \in A_r$ and
 $(p,u_0,q),(r,v_0,s) \notin M$, where $p$ (\resp $q$, $r$, $s$)  is
 the shared state $p_0$ of any
 path $(p_{i},u_i,p_{i+1})_{{-m}\leq i \leq a-1}$  
(\resp $(p_{i},u_i,p_{i+1})_{{-m+1}\leq i \leq {a}}$,
  $(p_{i},v_i,p_{i+1})_{{-m}\leq i \leq {a-1}}$,   $(p_{i},v_i,p_{i+1})_{_{-m+1}\leq i \leq {a}}$).
Then $\B(X)$  with the sets $F,G$ is a strictly-locally-Dyck testable.

Conversely, if there are nonnegative integers $m,a$ and sets $F,G$
such that $\B(X)$ is the set of
admissible words for $S(F,G)$ union the factors of length less than
$m+a+1$ of these words. Let $(\A,M)$ be the Dyck automaton obtained from the 
$(m,a)$-Dyck-De-Bruijn automaton by discarding the edges $((au,bv),b,(ub,vc))$
such that $aub \in F$ and setting that each edge $((au,bv),b,$ $(ub,vc))$ is
matched with $((a'u',b'v'),b',(u'b',v'c'))$ when $b \in A_c$, $b' \in
A_r$, and 
$(aubvc,a'u'b'v'c') \notin G$.
Then $(\A,M)$ is an $(m,a)$-local Dyck automaton accepting $X$. 
Thus $X$ is a finite-type-Dyck shift.
\end{proof}

\section{Edge-Dyck shifts} \label{section.edgeDyck}

A  \emph{Dyck graph} $(\G=(Q,E\subset Q \times Q),M)$
is composed of a graph $\G$, where the edges $E=(E_c,E_r,E_i)$ are partitioned into
three categories: the \emph{call edges} (denoted $E_c$), the 
\emph{return edges} $E_c$, and the \emph{internal edges} (denoted
$E_i$). 
The set $M$ is a set of pairs $(e,f)$, with $e \in E_c$, $f \in E_r$
of matched edges.  A finite path in $(\G,M)$ is admissible 
if it is admissible in the Dyck automaton $(\A,M)$ obtained from
$(\G,M)$ by labeling an edge of $\G$ by itself. An infinite path
is admissible if all its finite factor are admissible.

An \emph{edge-Dyck shift}  is the set of admissible bi-infinite paths
of a Dyck graph. The edge-Dyck shift defined by the Dyck graph $(G,M)$
is denoted by $\X_{(\G,M)}$. 

A Dyck graph is \emph{essential} if each edge and each matched pair of
edges are essential.

The following proposition shows that each finite-type Dyck shifts are
properly conjugate to an edge-Dyck shift.
\begin{proposition}
Each finite-type-Dyck shift is properly conjugate to a finite-type
edge-Dyck shift.
\end{proposition}

\begin{proof}
Let $(\A=(Q,E,A),M)$ be an $(m,a)$-local Dyck automaton accepting the 
finite-type-Dyck shift $X$ over $A=(A_c,A_r,A_i)$. 
Without loss of generality, we may
assume that for any pair of states $p,q$, there is at most one edge
going from $p$ to $q$. Otherwise we build the automaton 
$(\A'=(Q\times E,E',A),M')$ where the states are pairs $(p,e)$ where $p
\in Q$ and $e $ is an edge coming in $p$. There is an edge
$((p,e),a,(q,f))$ in $E'$ whenever $f=(p,a,q) \in E$. The pair of the two edges
$((p,e),a,(q,f))$ and $((r,g),b,(s,g))$ belongs to $M'$ whenever 
the pair $(p,a,q)$, $(r,b,s)$ belongs to $M$.
The Dyck-shift $X$ is still accepted by $(\A',M')$.

Let $E=(E_c,E_r,E_i)$ be the alphabet of call edges, return edges and
internal edges of $E$. Let $N$ be the set of pairs $(p,q),(r,s)$ such
that there is $a \in A_c$ and $b \in A_r$ and $(p,a,q),(r,b,s) \in M$.
The Dyck graph $(\G=(Q,E),N)$ over the alphabet $E$ defines an edge-Dyck shift $Y$.
There is a proper $(0,0)$-local block map $\Phi$ from $Y$ to $X$ 
defined by $\phi(p,a,q)=a$. The $(m,a)$-block map $\Psi:X \rightarrow
Y$ defined by $\psi(aubvc) = (p,b,q)$, where $p$ (\resp $q$) is the state $p_0$ (\resp $q_0$) of any path labelled by
  $aubv$ (\resp $ubvc$) going from $p_{-m}$ to $p_a$ 
(\resp going from $q_{-m}$ to $q_a$), is the inverse of $\Phi$.
\end{proof}




The notion of Dyck splittings of Dyck graphs is a particular case of
Dyck splittings of Dyck automata.

Let $(\G=(Q,E),M)$ be a Dyck automaton over $E=(E_c,E_r,E_i)$.
Let $p \in Q$ and $\P$ a partition $(\P_1,\ldots,\P_k)$ of size $k$ of
the edges coming in $p$. We define
a Dyck automaton $(\G'=(Q',E'),M')$ by 
\begin{itemize}
\item $Q'= Q \setminus \{p\} \cup \{p_1,\ldots,p_k\}$,
\item $(q,r) \in E'_c$ (\resp $E'_r$, $E'_i$) if $q,r \neq p$ and $(p,r) \in E_c$,
\item $(q,p_i) \in E'$ for each $1 \leq i \leq k$ such that
 $(q,p) \in \P_i$,
\item $(p_i,r) \in E'$ 
for each $1 \leq i \leq
  k$ such that  $(p,r) \in E$.
\item $M'$ is the set of pairs of edges $(q,r),(s,t)$ such that
  $(\pi(q),a,\pi(r)),$ $(\pi(s),b,\pi(t))$ $\in M$
where $\pi(q) = q$ for $q \neq p$ and $\pi(p_i)=p$ for $1 \leq i \leq k$.
\end{itemize}
This Dyck graph $(\G',M')$ is called a \emph{Dyck in-split graph} of
$(\A,M)$ and $(\A,M)$ is called a \emph{Dyck in-amalgamation graph} of
$(\A',M')$.

\begin{figure}[htbp]
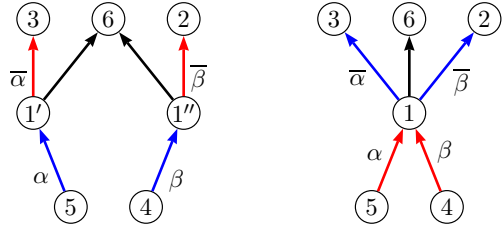

    \centering
\FixVCScale{0.5}
\VCDraw{%
\begin{VCPicture}{(0,-1)(12,5)}
\SetEdgeArrowWidth{8pt}
\MediumState
\SetEdgeLineWidth{2pt}
\State[1']{(0,2)}{11}
\State[1'']{(4,2)}{00}
\State[2]{(4,4.5)}{22}
\State[6]{(2,4.5)}{66}
\State[3]{(0,4.5)}{33}
\State[4]{(3,-0.5)}{44}
\State[5]{(1,-0.5)}{55}
\SetEdgeLineColor{red}
\EdgeL[0.3]{11}{33}{\overline{\alpha}}
\EdgeR[0.3]{00}{22}{\overline{\beta}}
\SetEdgeLineColor{blue}
\EdgeR[0.3]{44}{00}{\beta}
\EdgeL[0.3]{55}{11}\alpha{}
\SetEdgeLineColor{black}
\EdgeL[0.3]{11}{66}{}
\VCPut{(8,0)}{
\MediumState
\State[1]{(2,2)}{1}
\State[2]{(4,4.5)}{2}
\State[6]{(2,4.5)}{6}
\State[3]{(0,4.5)}{3}
\State[4]{(3,-0.5)}{4}
\State[5]{(1,-0.5)}{5}
\SetEdgeLineColor{blue}
\EdgeR[0.5]{1}{2}{\overline{\beta}}
\EdgeL[0.5]{1}{3}{\overline{\alpha}}
\SetEdgeLineColor{red}
\EdgeR[0.5]{4}{1}{\beta}
\EdgeL[0.5]{5}{1}{\alpha}
\SetEdgeLineColor{black}
\EdgeL[0.5]{1}{6}{}
\EdgeR[0.3]{00}{66}{}
\RstEdgeLineWidth
}
\end{VCPicture}%
        }
        \caption{A trim Dyck in-amalgamation of $(\G',M')$ to $(\G,M)$.
The blue edges are call edges, the red edges are return edges and 
the black edges are internal edges. 
In $\G'$ (on the left) the pairs of matched edges are $(5,1')(1',3)$ and
$(4,1")(1",2)$. In $\G$ (on the right), the pairs of matched edges are $(5,1)(1,3)$ and
$(4,1)(1,2)$. 
Call edges are labelled in $\Sigma$ and return
edges are labelled in $\mathfrak{P}(\overline{\Sigma})$, where $\Sigma$ is a
finite alphabet and $\overline{\Sigma}$ is a disjoint copy of $\Sigma$.
A call edge $(p,\alpha,q)$ and a return edge $(r,u,s)$ are matched 
if and only if $\overline{\alpha} \in u$.
Note that $(5,1)(1,2)$ and $(4,1)(1,3)$ are not
admissible and thus the amalgamation defines a proper conjugacy
between the admissible bi-infinite paths of $\G'$ and $\G$. }\label{figure.dyckStateMerging}
\end{figure}

The following proposition shows that Dyck in-amalgamation maps
commute. A similar result holds for Dyck out-amalgamation maps.

\begin{proposition} \label{proposition.commutation} Let
  $\X_{(\G_1,M_1)}$ be an edge-Dyck shift. Let $\Phi$ (\resp $\Psi$ be
  a Dyck
in-amalgamation transforming $(\G_1,M_1)$ to  $(\G_2,M_2)$ (\resp
$(\G_3,M_3)$. Then there is a Dyck graph $(\G_4,M_4)$ and a
Dyck in-amalgama\-tion $\Omega$ (\resp $\Theta$) from $\X_{(\G_2,M_2)}$ (\resp $\X_{(\G_3,M_3)}$)
to $\X_{(\G_4,M_4)}$ such that $\Phi \circ \Omega = \Psi \circ \Theta$.
\end{proposition}

\begin{figure}[htbp]
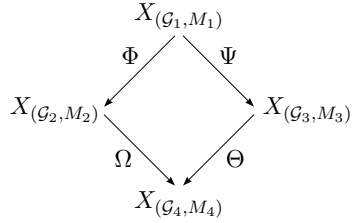

\SmallPicture%
\newlength{\dwd}\newlength{\dhg}%
\setlength{\dwd}{2.1cm}%
\setlength{\dhg}{2.1cm}%
\newlength{\dhgPlus}
\setlength{\dhgPlus}{2.3cm}%
\centering%
\VCDraw{%
\begin{VCPicture}{(0,-0.3)(2\dwd,5)}
\ChgStateLineColor{white}
\State[X_{(\G_4,M_4)}]{(1\dwd,-0.4)}{L01}
\State[X_{(\G_2,M_2)}]{(-1.2,\dhg)}{L11}
\State[X_{(\G_3,M_3)}]{(2.4\dhgPlus,\dhg)}{L12}
\State[X_{(\G_1,M_1)}]{(\dwd,2\dhgPlus)}{L21}
\HideState
\RstStateLineColor
\VSState{(1\dwd,0)}{01}
\State[A]{(1\dwd,0)}{L01}
\VSState{(0,\dhg)}{11}
\VSState{(2\dwd,\dhg)}{12}
\State[B]{(0,\dhg)}{L11}
\State[C]{(2\dwd,\dhg)}{L12}
\VSState{(\dwd,2\dhg)}{21}
\State[D]{(\dwd,2\dhg)}{L21}
\EdgeR{21}{11}{\Phi}
\EdgeL{21}{12}{\Psi}
\EdgeR{11}{01}{\Omega}
\EdgeL{12}{01}{\Theta}
\end{VCPicture}}
\caption{The commutation of Dyck in-amalgamation maps. If $\X_{(\G_2,M_2)},X_{(\G_3,M_3)}$ are edge-Dyck
  shifts which are Dyck in-amalgamations of $\X_{(\G_1,M_1)}$, then there is an
  edge-Dyck 
  shift $\X_{(\G_4,M_4)}$ which is a common Dyck in-amalgamation shift
  of $\X_{(\G_2,M_2)}$
  and $\X_{(\G_3,M_3)}$.}
\label{figure.commutation}
\end{figure}


\begin{proof}
Suppose
  that $\X_i=X_{(\G_i,M_i)}$ where $\G_i=(Q_i,E_i)$ for
  $i=1,2,3$. 
Let us assume that 
  the states $p_1, \ldots, p_k$ of $Q_1$ are amalgamated to a
  state $p$ of $Q_2$ and that the vertices $q_1, \ldots, q_{\ell}$
  of $Q_{1}$ are amalgamated to a vertex $q$ of $Q_{3}$.

  Let us first assume that the vertices $p_1, \ldots, p_{k}$ and
  $q_1, \ldots, q_{\ell}$ are all distinct. We define $(\G_4,M_4)$ as the
  in-amalgamation of $(\G_2,M_2)$ obtained by amalgamating the vertices $p,q_1,
  \ldots, q_{\ell}$ to a vertex $q$. It is also equal to the in-amalgamation of
  $\X_3$ obtained by amalgamating the vertices $q, p_1, \ldots, p_{k}$ to
  a vertex $p$.  
Let us now assume that 
  $p_1= q_1, \ldots, p_n=q_n$ for some
  integer $1\leq n \leq \min(k,\ell)$.  
This implies the following properties:
\begin{itemize}
\item   for any $1 \leq j \leq k$, $1 \leq j' \leq \ell$, one has
$(p_j,r) \in E_{1,c}$  (\resp $E_{1,r}$, $E_{1,i}$) if and only if $(p_{j'},r) \in E_{1,c}$  (\resp $E_{1,r}$, $E_{1,i}$) for $1\leq j,j'
\leq k$.
\item $(r,p_j) \in E_1$ implies $(r,q_{j'}) \notin E_1$ for any $j \neq j'$. 
\item  $(r,s)(p_j,t) \in M_1$ if and only if $(r,s)(q_{j'},t) \in
M_1$
\item $(p_j,t)(r,s) \in M_1$ if and only if $(q_{j'},q)(r,s) \in M_1$
for $1 \leq j \leq k$, $1 \leq j' \leq \ell$.
\end{itemize}
 We define $(\G_4,M_4)$ as the
  Dyck in-amalgamation of $(\G_2,M_2)$ obtained by the map $\Omega$ amalgamating the vertices $p, q_{n+1},
  \ldots, q_{\ell}$ to the state $p$. It is equal to the Dyck in-amalgamation of
  $(\G_3,M_3)$ obtained by the map $\Theta$ amalgamating the states $q, p_{n+1},$ $\ldots, p_{k}$ to
  a vertex $p$. The maps $\Omega$ and $\Theta$ are Dyck
  in-amalgamations and $\Phi \circ \Omega = \Psi \circ \Theta$.
\end{proof}





Unfortunately, two trim Dyck in-amalgamations do not commute in
general as is shown
in the following example.
\begin{example}
Let $(\G,M)$ be the Dyck graph of Figure \ref{figure.dyckGraphG}. Two trim
in-amalgamations are possible. States $1$ and $2$ are in-amalgamated in 
the Dyck graph $(\G_1,M_1)$ on the right of Figure
\ref{figure.dyckGraphG1G2} and states $2$ and $3$ are in-amalgamated in 
the Dyck graph $(\G_2,M_2)$ on the left of Figure
\ref{figure.dyckGraphG1G2} through a trim in-amalgamation.

\begin{figure}[htbp]
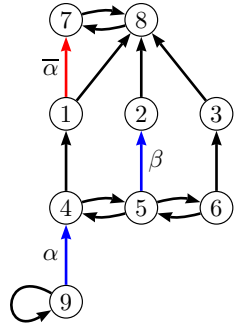

    \centering
\FixVCScale{0.5}
\VCDraw{%
\begin{VCPicture}{(0,-3)(5,5)}
\SetEdgeArrowWidth{8pt}
\MediumState
\SetEdgeLineWidth{2pt}
\State[1]{(0,2)}{11}
\State[2]{(2,2)}{22}
\State[3]{(4,2)}{33}
\State[4]{(0,-0.5)}{44}
\State[5]{(2,-0.5)}{55}
\State[6]{(4,-0.5)}{66}
\State[7]{(0,4.5)}{77}
\State[8]{(2,4.5)}{88}
\State[9]{(0,-3)}{99}
\SetEdgeLineColor{red}
\EdgeL[0.5]{11}{77}{\overline{\alpha}}
\SetEdgeLineColor{blue}
\EdgeR[0.5]{55}{22}{\beta}
\EdgeL[0.5]{99}{44}{\alpha}
\SetEdgeLineColor{black}
\LoopW{99}{}
\EdgeL[0.5]{44}{11}{}
\EdgeR[0.5]{66}{33}{}
\ArcL[0.5]{44}{55}{}
\ArcL[0.5]{55}{44}{}
\ArcL[0.5]{66}{55}{}
\ArcL[0.5]{55}{66}{}
\ArcL[0.5]{77}{88}{}
\ArcL[0.5]{88}{77}{}
\EdgeL[0.5]{11}{88}{}
\EdgeL[0.5]{22}{88}{}
\EdgeL[0.5]{33}{88}{}
\end{VCPicture}%
        }
        \caption{A Dyck graph $(\G,M)$.}\label{figure.dyckGraphG}
\end{figure}

\begin{figure}[htbp]
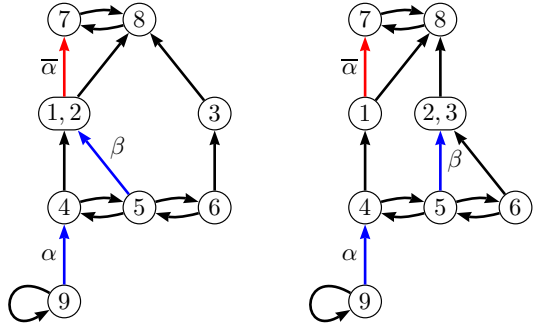

    \centering
\FixVCScale{0.5}
\VCDraw{%
\begin{VCPicture}{(0,-3)(12,5)}
\SetEdgeArrowWidth{8pt}
\MediumState
\SetEdgeLineWidth{2pt}
\StateVar[1,2]{(0,2)}{11}
\State[3]{(4,2)}{33}
\State[4]{(0,-0.5)}{44}
\State[5]{(2,-0.5)}{55}
\State[6]{(4,-0.5)}{66}
\State[7]{(0,4.5)}{77}
\State[8]{(2,4.5)}{88}
\State[9]{(0,-3)}{99}
\SetEdgeLineColor{red}
\EdgeL[0.5]{11}{77}{\overline{\alpha}}
\SetEdgeLineColor{blue}
\EdgeR[0.5]{55}{11}{\beta}
\EdgeL[0.5]{99}{44}{\alpha}
\SetEdgeLineColor{black}
\LoopW{99}{}
\EdgeL[0.5]{44}{11}{}
\EdgeR[0.5]{66}{33}{}
\ArcL[0.5]{44}{55}{}
\ArcL[0.5]{55}{44}{}
\ArcL[0.5]{66}{55}{}
\ArcL[0.5]{55}{66}{}
\ArcL[0.5]{77}{88}{}
\ArcL[0.5]{88}{77}{}
\EdgeL[0.5]{11}{88}{}
\EdgeL[0.5]{33}{88}{}
\VCPut{(8,0)}{
\MediumState
\State[1]{(0,2)}{1}
\StateVar[2,3]{(2,2)}{2}
\State[4]{(0,-0.5)}{4}
\State[5]{(2,-0.5)}{5}
\State[6]{(4,-0.5)}{6}
\State[7]{(0,4.5)}{7}
\State[8]{(2,4.5)}{8}
\State[9]{(0,-3)}{9}
\SetEdgeLineColor{red}
\EdgeL[0.5]{1}{7}{\overline{\alpha}}
\SetEdgeLineColor{blue}
\EdgeR[0.5]{5}{2}{\beta}
\EdgeL[0.5]{9}{4}{\alpha}
\SetEdgeLineColor{black}
\LoopW{9}{}
\ArcL[0.5]{4}{5}{}
\ArcL[0.5]{5}{4}{}
\ArcL[0.5]{6}{5}{}
\ArcL[0.5]{5}{6}{}
\ArcL[0.5]{7}{8}{}
\ArcL[0.5]{8}{7}{}
\EdgeL[0.5]{1}{8}{}
\EdgeL[0.5]{2}{8}{}
\EdgeL[0.5]{4}{1}{}
\EdgeR[0.5]{6}{2}{}
\RstEdgeLineWidth
}
\end{VCPicture}%
        }
        \caption{Two trim Dyck in-amalgamations of $(\G,M)$. States $1$ and $2$ are in-amalgamated in 
the Dyck graph $(\G_1,M_1)$ (on the right of the figure) and states $2$ and $3$ are in-amalgamated in 
the Dyck graph $(\G_2,M_2)$ through a trim in-amalgamation (on the left of the figure).}\label{figure.dyckGraphG1G2}
\end{figure}

\begin{figure}[htbp]
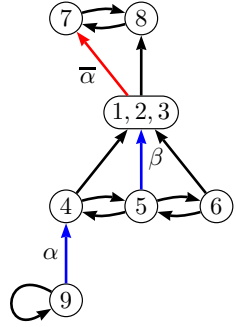

    \centering
\FixVCScale{0.5}
\VCDraw{%
\begin{VCPicture}{(0,-3)(5,5)}
\SetEdgeArrowWidth{8pt}
\MediumState
\SetEdgeLineWidth{2pt}
\StateVar[1,2,3]{(2,2)}{22}
\State[4]{(0,-0.5)}{44}
\State[5]{(2,-0.5)}{55}
\State[6]{(4,-0.5)}{66}
\State[7]{(0,4.5)}{77}
\State[8]{(2,4.5)}{88}
\State[9]{(0,-3)}{99}
\SetEdgeLineColor{red}
\EdgeL[0.5]{22}{77}{\overline{\alpha}}
\SetEdgeLineColor{blue}
\EdgeL[0.5]{99}{44}{\alpha}
\EdgeR[0.5]{55}{22}{\beta}
\SetEdgeLineColor{black}
\EdgeL[0.5]{44}{22}{}
\EdgeR[0.5]{66}{22}{}
\ArcL[0.5]{44}{55}{}
\ArcL[0.5]{55}{44}{}
\ArcL[0.5]{66}{55}{}
\ArcL[0.5]{55}{66}{}
\ArcL[0.5]{77}{88}{}
\ArcL[0.5]{88}{77}{}
\LoopW{99}{}
\EdgeL[0.5]{22}{88}{}
\end{VCPicture}%
        }
        \caption{The Dyck graph $(\G_3,M_3)$.}\label{figure.dyckGraphG3}
\end{figure}
The Dyck graphs $(\G_1,M_1)$ and $(\G_2,M_2)$ cannot be (trim) in-amalgamated
to the Dyck graph $(\G_3,M_3)$ of Figure \ref{figure.dyckGraphG3}. For instance the states $(1,2)$ and
$3$ of $\G_1$ cannot be in-amalgamated since the bijection between
bi-infinite admissible paths of $(\G_1,M_1)$ and $(\G_3,M_3)$ would be
lost. Indeed, the path
\begin{equation*}
\dotsm 9 \rightarrow 9 \xrightarrow{\alpha} 4 \rightarrow 5
\rightarrow (1,2,3) \xrightarrow{\overline{\alpha}} 7 \rightarrow 8 \dotsm
\end{equation*}
is an admissible path of $(\G_3,M_3)$ which is not the image of an
admissible path of $(\G_1,M_1)$ by the
$(0,0)$-block amalgamation map.
\end{example}

\subsection{Decomposition Theorem} \label{section.decompostion}

The Decomposition Theorem for shifts of infinite words states that
any conjugacy between shifts of finite type can be decomposed into a
finite sequence of splittings and amalgamations (see for
instance~\cite{Kitchens1998}). In this section, we prove a
similar result for proper conjugacies between edge-Dyck
shifts. As for shifts of finite type, the crucial lemma will show that
the memory of a block map can be reduced using proper
splittings.

\begin{theorem} \label{theorem.decomposition}
Let $(G,M)$, $(\H,N)$ be two Dyck graphs such that
$\X_{(\G,M)}$ and $\X_{(\H,N)}$ are properly conjugate.
Then there are sequences of Dyck graphs
$(\G_i,M_i)_{1 \leq i \leq k}$, $(\H_j,N_j)_{1 \leq j \leq r}$
and Dyck (or trim Dyck) in-splittings  $\Psi_i: (\G_i,M_i) \rightarrow
(\G_{i+1},M_{i+1})$, $\Delta_j:  (\H_j,N_j) \rightarrow
(\H_{j+1},N_{j+1})$, for $1 \leq i \leq k-1$, $1 \leq j \leq k'-1$
such that $(\G_1,M_1) = (\G,M)$, $(\H_1,N_1) = (\H,N)$,
and $(\G_k,M_k) = (\H_{k'},N_{k'})$, up to renaming of the states.
\begin{equation*}
(\G,M) \xrightarrow{\Psi_1} \ldots
\xrightarrow{\Psi_k}  (\G_k,M_k) = (\H_{k'},N_{k'}) 
\xleftarrow{\Delta_{k'}}  \ldots
\xleftarrow{\Delta_{1}} (\H,N) 
\end{equation*}
\end{theorem}

We also obtain the following corollary.
\begin{corollary} \label{theorem.decomposition2}
Any proper conjugacy between edge-Dyck shifts can be decomposed
as a composition of Dyck in-splitting maps and Dyck in-amalga\-mation maps.
\end{corollary}

\begin{lemma} \label{lemma.d1} Let $\Phi : X_{(\G,M)} \rightarrow X_{(\H,N)}$ be a proper
  $(m,a)$-block conjugacy between two edge-Dyck shifts with $m \geq 1$. Then
  there is a Dyck in-split graph $(\widetilde{G},\widetilde{M})$ of 
$(\G,M)$ and a Dyck in-splitting map $\Psi_1$ from $\X_{(\G,M)}$ to
  $\X_{(\widetilde{G},\widetilde{M})}$, and a proper $(m-1,a)$-block conjugacy $\widetilde{\Phi}$ from
  $\X_{(\widetilde{G},\widetilde{M})}$ onto $\X_{(\H,N)}$
  such that $\Phi = \widetilde{\Phi} \circ
  \Psi_1$.
\end{lemma}

\begin{proof}
Let $E=(E_c,E_r,E_i)$ (\resp $F=(F_c,F_r,F_i)$) the edges of $X=X_{(\G,M)}$ (\resp $Y=X_{(\H,N)}$).
Let $\phi: \mathcal{B}_{m+a}(X) \rightarrow F$ be the local function
defining~$\Phi$.
Let $\X_{(\widetilde{G},\widetilde{M})}$ be the edge-Dyck graph
obtained by splitting each state $p$ into states $(p,e)$, where $e \in
\In(p)$, according to the trivial partition of the edges $\In(p)$
coming in $p$ where each class is a singleton. 
Hence we perform a full Dyck in-splitting.
There is an edge 
$((q,f),(p,e))$ in $\widetilde{G}$ if and only if $f=(r,q)$ and
$e=(q,p)$ for some $r \in Q$. 

Let $\Psi_1 : X_{(\G,M)} \rightarrow X_{(\widetilde{G},\widetilde{M})}$ be the $(1,0)$-block conjugacy
defined by $\psi_1(fe)= ((q,f)(p,e))$. 

Let $\widetilde{\Phi} : X_{(\widetilde{G},\widetilde{M})} \rightarrow Y$ be the $(m-1,a)$-block map defined, for
  any block of $\widetilde{E}^{m-1+a}$, by
 \begin{equation*}
\widetilde{\phi} (\tilde{e}_{-m+1} \cdots \tilde{e}_0\cdots
\tilde{e}_{a}) = \phi(e_{-m} \cdots e_0\cdots e_{a}),
\end{equation*}
where $\tilde{e}_i=((p_{i-1},e_{i-1}),(p_i,e_i))$.
 We have $\Phi =
  \widetilde{\Phi} \circ \Psi_1 $.
\end{proof}

A similar result holds for reducing the anticipation of $\Phi$.

\begin{lemma} \label{lemma.d2} Let $\Phi : X_{(\G,M)} \rightarrow  X_{(\H,N)}$ be a proper
  $(0,0)$-block conjugacy between two edge-Dyck shifts such that $\Phi^{-1}$ is
  an $(m,a)$-block map with $m \geq 1$.  
 Then
  there is a Dyck in-split graph $(\widetilde{G},\widetilde{M})$ of 
$(\G,M)$, a Dyck in-split graph $(\widetilde{H},\widetilde{N})$ of 
$(\H,N)$, a Dyck in-splitting map
  $\Psi_1$ from $\X_{(\G,M)}$ to $\X_{(\widetilde{G},\widetilde{M})}$, a Dyck in-splitting map $\Psi_2$
  from $\X_{(\H,N)}$ to $\X_{(\widetilde{H},\widetilde{N})}$, and a $(0,0)$-block conjugacy
  $\widetilde{\Phi}$ from $\X_{(\widetilde{G},\widetilde{M})}$ onto
  $\X_{(\widetilde{H},\widetilde{N})}$ such
  that $\Phi = \Psi_2^{-1} \circ \widetilde{\Phi} \circ \Psi_1$ and
  $\widetilde{\Phi}^{-1}$ is a proper $(m-1,a)$-block map.  This makes the
    following diagram commute.
\begin{equation*}
\begin{CD}
X_{(\G,M)} @>\Phi>> X_{(\H,N)} \\
@V\Psi_1VV  @VV \Psi_2 V \\
X_{(\widetilde{G},\widetilde{M})} @>\widetilde{\Phi}>>  X_{(\widetilde{H},\widetilde{N})}
\end{CD}
\end{equation*}
\end{lemma}

\begin{proof}
Let $E=(E_c,E_r,E_i)$ (\resp $F=(F_c,F_r,F_i)$) the edges of $X=X_{(\G,M)}$ (\resp $Y=X_{(\H,N)}$).
Let $\phi: \mathcal{B}_{m+a}(X) \rightarrow F$ be the local function
defining~$\Phi$.
Let $\X_{(\widetilde{G},\widetilde{M})}$ be the edge-Dyck graph
obtained by a Dyck in-splitting of each state $p$ into states $(p,[e])$, where $e \in
\In(p)$, according to the partition of the edges 
coming in $p$ such that two edges $e,e'$ belong to a same class if and only
if $\phi(e)=\phi(e')$ where $\phi$ is the local function
associated to $\Phi$.  Let $[e]$ denotes the class of $e$ in this
partition. There is an edge 
$((q,[f]),(p,[e]))$ in $\widetilde{G}$ if and only if $f=(r,q)$,
$e=(q,p)$ for some $r \in Q$ and $fe$ is admissible.
Let $\Psi_1 : X \rightarrow X_{(\widetilde{G},\widetilde{M})}$ 
  be the $(1,0)$-block map defined by $\psi_1(fe) = ((q,[f])(p,[e]))$. 

Let $\X_{(\widetilde{H},\widetilde{N})}$ be the edge-Dyck graph
obtained by a Dyck in-splitting of each state $p$ into states $(p,e)$, where $e \in
\In(p)$, according to the trivial partition of the edges $\In(p)$
coming in $p$ where each class is a singleton. There is an edge 
$((q,f),(p,e))$ in $\widetilde{H}$ if and only if $f=(r,q)$,
$e=(q,p)$ for some $r \in Q$, and $fe$ is admissible. We
  denote by $\Psi_2$ the proper in-splitting map from $Y$ to
  $\X_{(\widetilde{H},\widetilde{N})}$.

  We define a $(0,0)$-block map $\widetilde{\Phi}$ from
  $\X_{(\widetilde{G},\widetilde{M})}$ onto $\X_{(\widetilde{H},\widetilde{N})}$
by $\widetilde{\phi}([e]) =\phi(e)$. It is consistent by
definition of the partition of $\In(p)$ if $e$ ends in $p$.
We have $\Phi = \Psi_2^{-1} \circ \widetilde{\Phi}
  \circ \Psi_1$. It remains to check that $\widetilde{\Phi}^{-1} =
  \Psi_1 \circ \Phi^{-1} \circ \Psi_2^{-1}$ is an $(m-1,a)$-block map.
  That is, we must show that for any word $x$ in $\X_{(\widetilde{H},\widetilde{N})}$, the 
  coordinate of index 0 of  
  $\widetilde{\Phi}^{-1}(x)$ is determined by the its block 
  $u=[x_{-m}x_{m+1}][x_{m+1}x_{m+2}]
\cdots [x_{a-1}x_a]$ of length $m+a$. But this follows from the
  observation that the block $u$ determines $\Psi_2^{-1}(x)_{-m}$
  and therefore the block $\Psi_2^{-1}(x)[-m,a]$ of length $m+a+1$.
Hence, if $x'= (\Phi^{-1} \circ \Psi_2^{-1})(x)$,
  $x'_0$ is determined by $u$ since $\Phi^{-1}$ is an
  $(m,a)$-block map. Furthermore, since the block $u$ determines
  the block $\Psi_2^{-1}(x)[-m,a]$ of length $m+a+1$, it determines $\phi(x'_{-1})$
  and thus $\Psi_1(x')_0$ which depends only on $x'_{0}$,
  $\phi(x'_{-1})$.
\end{proof}
A similar result holds for reducing the anticipation of $\Psi_{-1}$.

\begin{lemma} \label{lemma.finalStep}
Let $\Delta: X_{(\G,M)} \rightarrow X_{(\H,N)}$ be a proper conjugacy
between edge-Dyck shifts defined by Dyck graphs
$(\G,M)$ and $(\H,N)$. Let us assume that $\Delta$ and $\Delta^{-1}$
are $(0,0)$-block maps. 
Then there is a Dyck graph $(\G',M')$
(\resp $(\H',N')$) obtained by trim in-splittings of $(\G,M)$ 
(\resp $(\H,N)$) such that $(\G',M')$ and
$(\H',N')$ are equal, up to a renaming of the states.
\end{lemma}

\begin{proof}
Let $\G=(Q,E)$ and $\H=(R,F)$.
Let $(\G'=(Q',E'),M')$ be the Dyck graph whose states are $(e,f)$ where 
$ef$ is a path of $E$ and edges are $((e,f)(f,g))$
if $efg$ is a path of $\G$.
The edge $((e,f),(f,g))$ is a call (\resp return, internal) edge if
and only if  $f$ is a call (\resp return, internal) edge. A return
edge $((e,f),(f,g))$ is matched with a call edge $((r,s),(s,t))$  
if and only if $f$ is matched with $s$. We define $(\H'=(R',F'),N')$
similarly for $(\H,N)$. 

The Dyck graph $(\G',M')$ (\resp $(\H',N')$) is obtained from 
$(\G,M)$ (\resp  $(\H,N)$) by trim in-splittings.
Each state $p$ is split into the states $(p,e)$ for each edge $e$
coming in $p$ with a trim in-splitting and each state $(p,e)$
is split into $(e,f)$ for each edge going out of $p$ (or $(p,e)$)
with a trim out-splitting. Since the splittings trim, 
for each edge $f$ of $(\G',M')$ (\resp $(\H',N')$) there is a bi-infinite admissible path
extending $f$ in $(\G',M')$ (\resp $(\H',N')$). 
The $(0,0)$-block proper conjugacy $\Delta':X_{(\G',M')} \rightarrow X_{(\H',N')}$
defined by $\delta'((e,f)(f,g))=
(\delta(e),\delta(f))(\delta(f),\delta(g))$ has a inverse which is
also a $(0,0)$-block map.

Let us show that $(\G',M')$ and
$(\H',N')$ are equal, up to a renaming of the states.
We define a renaming of the states $\rho: Q' \rightarrow R'$ as follows. 
If $(e,f)$ is a state of $\G'$, we set $\rho(e,f)=
(\delta(e),\delta(f))$. If $((e,f)(f,g))$ is an edge of $E'$, there
is a bi-infinite admissible path $\pi=((e_i,e_{i+1})(e_{i+1},e_{i+2}))_{i \in \Z}$ of $\G'$
such that $(e_0,e_1)=(e,f)$ and $(e_1,e_2)=(f,g)$.
This image of $\pi$ by $\Delta'$ is the admissible path 
$((\delta(e_i),\delta(e_{i+1}))(\delta(e_{i+1}),\delta(e_{i+2})))_{i
  \in \Z}$. Hence $(\rho(e,f),\rho(f,g))$ is an
edge of $F'$.  
Inverting the roles played by $(\G',M')$ and
$(\H',N')$, we obtain that $\rho$ is a graph isomorphism
from $\G'$ into $\H'$. 

Furthermore, since $(\G',M')$ and $(\H',N')$ are essential, for each
pair $(e',f'),$ $(r',s')$ of matched edges of $(\G',M')$, there
is a bi-infinite admissible path $\pi$ such that $\pi= z (e',f')
w (r',s') z'$, where $w$ is a Dyck path. The image of $\pi$ by 
$\Delta'$ is the bi-infinite admissible path $\pi= \delta'(z)
\delta'(e',f') \delta'(w) \delta'(r',s') \delta'(z')$ of
$(\H',N')$. Since $\delta'(w)$ is a Dyck path of $(\H',N')$.
Hence $\delta'(e',f'),\delta'(r',s')$ belongs to $N'$.
We get $M' \subseteq N'$ and symmetrically $N' \subseteq M'$. 
Thus $(\G',M') = (\H',N')$.
\end{proof}

\begin{proof}[Proof of Theorem~\ref{theorem.decomposition}]
Let $\Phi:  X_{(\G,M)} \rightarrow X_{(\H,N)}$ be a proper $(m,a)$-block conjugacy
between two edge-Dyck shifts such that $\Phi^{-1}$ is an
$(m',a')$-block map where $m,a,m',a'$ are nonnegative integers.

  By Lemma~\ref{lemma.d1} and Lemma~\ref{lemma.d2}, there are 
 there are sequences of Dyck graphs
$((\G_i,M_i))_{1 \leq i \leq k}$, $((\H_j,N_j))_{1 \leq j \leq r}$
and Dyck in-splittings $\Psi_i: (\G_i,M_i) \rightarrow
(\G_{i+1},M_{i+1})$, $\Delta_j:  (\H_j,N_j) \rightarrow
(\H_{j+1},N_{j+1})$, for $1 \leq i \leq k-1$, $1 \leq j \leq k'-1$
such that $(\G_1,M_1) = (\G,M)$, $(\H_1,N_1) = (\H,N)$,
and a proper conjugacy $\Delta: X_{(\G_k,M_k)} \rightarrow
X_{(\H_{k'},N_{k'})}$ such $\Delta$ and $\Delta^{-1}$ are $(0,0)$-block maps,
with $\Phi = \Delta_1^{-1}\circ \Delta_2^{-1} \cdots \circ \Delta_{k'}^{-1} \circ
  \Delta \circ  \Psi_{k} \circ \Psi_2 \circ \Psi_1$.   

By Lemma \ref{lemma.finalStep}, 
there are trim Dyck in-splitting maps $(\Psi_{j})_{k+1 \leq j \leq n}$ splitting
$(\G_{j-1},M_{j-1})$ into an essential Dyck graph $(\G_{j},M_{j})$
and trim Dyck in-splitting maps $(\Delta_{j'})_{k'+1 \leq j' \leq n'}$ splitting
$(\H_{j'-1},N_{'-1})$ into an essential Dyck graph $(\H_{j'},N_{j'})$, 
and a proper conjugacy $\Delta': (\G_{n},M_{n})
\rightarrow (\H_{n'},N_{n'})$ such that 
$\Phi = \Delta_1^{-1}\circ \Delta_2^{-1} \cdots \circ \Delta_{n'}^{-1} \circ
  \Delta \circ  \Psi_{n} \circ \Psi_2 \circ \Psi_1$ and $\Delta$
is a renaming map.
\end{proof}
The Decomposition Theorem is illustrated in the following diagram where $\Delta$ is a renaming map.
\begin{equation*}
\begin{CD}
X_{(\G,M)} @>\Phi>> X_{(\H,N)}\\
@V\Psi_1VV  @VV \Delta_1 V \\
\vdots @. \vdots \\
@V\Psi_{n}VV  @VV \Delta_{n'} V \\
X_{(\G_n,N_n)} @>\Delta >> X_{(\H_{n'},N_{n'})}\\
\end{CD}
\end{equation*}

\begin{small}
\bibliographystyle{abbrv} 
\bibliography{soficDyck}
\end{small}

\end{document}